\documentclass[USenglish, cleveref]{lipics-v2021}
\usepackage[utf8]{inputenc}
\usepackage[T1]{fontenc}
\usepackage{lmodern}
\usepackage{microtype}
\hideLIPIcs
\nolinenumbers
\makeatletter
\renewcommand\@oddfoot{
	\hfil
	\rlap{%
		\vtop{%
			\vskip10mm
			\colorbox[rgb]{0.99,0.78,0.07}
			{\@tempdima\evensidemargin
				\advance\@tempdima1in
				\advance\@tempdima\hoffset
				\hb@xt@\@tempdima{%
					\textcolor{darkgray}{\normalsize\sffamily
						\bfseries\quad
						\expandafter\textsolittle\expandafter{
							arXiv.org}}%
					\strut\hss}}}}
}
\makeatother
\usepackage{xspace}
\usepackage{graphicx}
\usepackage{tikz}
\usepackage{amssymb,amsmath,mathtools,bm}
\hypersetup{hidelinks}
\usepackage{hyperref}
\usepackage{cancel}

\providecommand{\lowerboundHc}{lowerboundHc}
\providecommand{\upperboundHc}{upperboundHc}
\providecommand{\upperboundHcOnotation}{upperboundHcOnotation}
\providecommand{\upperboundHcYao}{upperboundHcYao}

\makeatletter
\newcommand{\laterdef}[2]{%
	\protected@write\@auxout{}{\gdef\string#1{#2}}}
\makeatother

\usepackage{color}
\definecolor{hellblau}{rgb}{0.6,0.6,0.8}
\newcommand{\cmt}[1]{\textcolor{hellblau}{\begin{array}{l}\text{#1}\end{array}\qquad}}

\newcommand{\ie}{that is,\xspace}

\newcommand{\algA}{\ensuremath{\textup{\textsc{Alg}}}\xspace}

\DeclareMathOperator{\cost}{cost}
\newcommand{\N}{\mathbb{N}}
\newcommand{\R}{\mathbb{R}}
\renewcommand{\P}{\mathbb{P}}
\newcommand{\Sn}{\mathcal{S}_n}

\newcommand{\euler}{\textnormal{\textrm{e}}}

\crefname{section}{Section}{Sections}
\crefname{lemma}{Lemma}{Lemmas}
\crefname{figure}{Figure}{Figures}
\crefname{theorem}{Theorem}{Theorems}
\crefname{definition}{Definition}{Definitions}
\crefname{corollary}{Corollary}{Corollaries}
\crefname{equation}{Equation}{Equations}

\let\leftold\left
\let\rightold\right
\renewcommand{\left}{\mathopen{}\mathclose\bgroup\leftold}
\renewcommand{\right}{\aftergroup\egroup\rightold}

\title{\boldmath Bounds for $c$-Ideal Hashing}
\titlerunning{\boldmath Bounds for $c$-Ideal Hashing}

\author{Fabian Frei}{Department of Computer Science, ETH Zürich}{fabian.frei@inf.ethz.ch}{https://orcid.org/0000-0002-1368-3205}{\thanks{Work done in part during a stay at Hosei University, supported by grant GR20109 by the Swiss National Science Foundation (SNSF) and the Japan Society for the Promotion of Science (JSPS).}}
\author{David Wehner}{Department of Computer Science, ETH Zürich}{david.wehner@inf.ethz.ch}{https://orcid.org/0000-0003-0201-4898}{}

\authorrunning{F.\,Frei and D.\,Wehner}

\Copyright{Fabian Frei and David Wehner}


\begin{document}

\maketitle

\begin{abstract}
In this paper, we analyze hashing from a worst-case perspective. To this end, we study a new property of hash families that is strongly related to $d$-perfect hashing, namely $c$-ideality. On the one hand, this notion generalizes the definition of perfect hashing, which has been studied extensively; on the other hand, it provides a direct link to the notion of $c$-approximativity . We focus on the usually neglected case where the average load $\alpha$ is at least $1$ and prove upper and lower parametrized bounds on the minimal size of $c$-ideal hash families.

As an aside, we show how $c$-ideality helps to analyze the \emph{advice complexity} of hashing. The concept of advice, introduced a decade ago, lets us measure the information content of an online problem. We prove hashing's advice complexity to be linear in the hash table size.
\end{abstract}

\section{Introduction}
Say you wake up one morning in a hotel with the desire to stroll around and view some sights of the city. 
Considering your peculiar proclivity towards paranoia, you want to know for every car you see whether you have seen this car before.
Whenever you see a car, you can quickly jot down the car plate and while doing so, you can check whether you have noted this number before.
You have to be quick though, since cars might follow each other closely and you will not have the time to go through your entire list to check whether a new car plate is in the list.
There are many possible car plates, maybe around $36^8$ for 8-character plates.
You estimate to see around 100 cars this day. You are not able to compare the new plate against more than $5$ car plates jotted down before. This is a situation where hashing excels.
However, this paper is not about normal hashing.

Luckily, you have heard of hashing and you know that when you first apply a hash function to the car plates---for example by choosing the last digit of the car plate---and then sort the strings by their hash value, you will be able to determine very quickly whether you have seen a given number before. But you are paranoid. Very paranoid. You know that you will only be able to quickly know you have seen a number before \emph{on average}.
You strive for perfection, however. Your hash function has to be good in the worst case as well.
You want to ensure you apply a hash function such that no hash value appears more than, say, $5$ times. You know this is not possible without further aid, but fret not, you have a very powerful friend who, provided you present him a small family of hash functions, is miraculously able to point you to the hash function among the family that achieves your goal---if there is a good hash function among your family, that is. And if the family you present is small. But, you wonder, does such a small family of hash functions exist at all? This paper deals with that question.

Hashing is one of the most popular tools in computing, both from a practical and a theoretical perspective.
Hashing was invented as an efficient method to store and retrieve information.
Its success began at latest in 1968 with the seminal paper ``Scatter Storage Techniques'' by Robert Morris~\cite{morris}, which was reprinted in 1983 with the following laudation~\cite{morris_reprint}:
\begin{quote}
	From time to time there is a paper which summarizes an emerging research area, focuses the central issues, and brings them to prominence. Morris's paper is in this class. [$\ldots$] It brought wide attention to hashing as a powerful method of storing random items in a data structure with low average access time.
\end{quote}
This quote illustrates that the practical aspects were the focus of the early research, and rightly so.
Nowadays, hashing has applications in many areas in computer science.
Before we describe how hashing can be applied, we state the general setting.

\subsection{General Setting and Notation}
We have a large set $U$, the \emph{universe}, of all possible elements, the \emph{keys}.
In our example, this would be the set of all possible car plates, so all strings of length $8$ over the alphabet $\{A,\ldots,Z, 0,\ldots,9\}$.
Then, there is a subset $S \subseteq U$ of keys from the universe. This set stands for the unknown elements that appear in our application. In our example, $S$ corresponds to the set of all car plates that we see on this day.
Then, there is a small set $T$, the \emph{hash table}, whose $m$ elements are called \emph{cells} or \emph{slots}. 
In our example, $T$ corresponds to our notebook and we organize the entries in our notebook according to the last character on the car plate, so each cell corresponds to a single letter of the alphabet.
Typically, the universe is huge in comparison to the hash table, for instance, $|U|=2^{|T|}$.
Every function from $U$ to $T$ is called a \emph{hash function}; it \emph{hashes} (i.e.,~maps) keys to cells.
We have to choose a function $h$ from $U$ to $T$ such that the previously unknown set $S$ is distributed as evenly as possible among the cells of our hash table.
For an introduction to hashing, one of the standard textbooks is the book by Mehlhorn~\cite{mehlhorn}.
We recommend the newer book by Mehlhorn and Sanders~\cite{mehlhorn_algorithmen}, and to German readers in particular its German translation by Dietzfelbinger, Mehlhorn and Sanders~\cite{dietzfelbinger_algorithmen}.

While $U$ and $T$ can be arbitrary finite sets, we choose to represent their elements by 
integers---which can always be achieved by a so-called pre-hash function---and
let $U\coloneqq\{1,\ldots,u\}$ and $T\coloneqq\{1,\ldots,m\}$.
For convenience, we abbreviate $\{1,\ldots,k\}$ by $[k]$ for any natural number $k$. 
We assume the size $n$ of the subset to be at least as large as the size of the hash table, that is, $|S|=n\ge m$ and, to exclude the corner case of hashing almost the entire universe, also $|U|\ge n^2$.

\subsection{Applications of Hashing}
Among the numerous applications of hashing, two broad areas stand out. First, as we have seen, hashing can be used to store and retrieve information. With hashing, inserting new information records, deleting records, and searching records can all be done in expected constant time, \ie, the number of steps needed to insert, find, or delete an information record does not depend on the size of $U$, $S$, or $T$.
For an example of such an application, consider dictionaries in Python, which are implemented using a hash table.
To create a dictionary, you specify pairs of keys (e.g., scientists) and values (e.g., their Erd\H{o}s number).
Now, retrieving values by calling the key, inserting new key-value-pairs, and deleting key-value-pairs can all be done in time independent of the number of scientists.
Another example is the Rabin-Karp algorithm~\cite{rabinkarp}, which searches a pattern in a text and can be used for instance to detect plagiarism in doctoral theses.
The naïve approach is to slide the pattern $p$, which has length $m$, letter by letter over the text $t$, which has length $n$, and then compare the pattern with each substring $s$ of length $m$ of the text.
Here, the universe $U$ consists of all possible texts (in a certain language) and $S$ are all texts of length $n$.
This needs $n-m+1$ comparisons between pattern and substrings and for each comparison, $m$ letters have to be compared. 
In total, there are $n(n-m+1)$ comparisons, which is a large computational effort.
A much faster approach in expectation is to use the Rabin-Karp algorithm.
Here, the hash values of the pattern and all substrings of length are computed first.
For a hash function with the property that the hash of a string, say $s_2 = 23456$, can easily be derived from the hash of the  previous string, say $s_1 = 12345$, this is computationally not too expensive.
Then, the hash value of the pattern is compared to the hash value of each substring.
So far, this is not an improvement. Now, however, the pattern and a substring are only compared if their hash values match.
This leads to an average time complexity of $O(n+m)$, which is much better than the previous $O(nm)$.

The second main application area is cryptography.
In cryptographic hashing, the hash function has to fulfill additional requirements, for example that it is computationally very hard to reconstruct a key $x$ from its hash $h(x)$ or to find a colliding key, \ie, a $y\in U$ with $h(y)=h(x)$.
Typical use cases here are, for example, digital signatures and famous examples of cryptographic hash functions are the checksum algorithms such as MD5 or SHA, which are often used to check whether two files are equal.
Moreover, and perhaps particularly interesting for computer science, hashing is a useful tool in the design of randomized algorithms and their derandomization. Further examples and more details can be found for example in the useful article by Luby and Wigderson~\cite{derandomization}.

\subsection{Theory of Hashing}
Soon after the early focus on the practical aspect, a rich theory on hashing evolved.
In this theory, randomization plays a pivotal role.
From a theory point of view, we aim for a hash function $h$ that reduces collisions among the yet-to-be-revealed keys of $S$ to a minimum.
One possibility for selecting $h$ is to choose the image for each key in $U$ uniformly at random among the $m$ cells of $T$. We can interpret this as picking a random $h$ out of the family $\mathcal{H}_\text{all}$ of all hash functions.\footnote{Note that the number of hash functions, $|\mathcal{H}_\text{all}|=m^u$, is huge.} 
Then the risk that two keys $x,y\in U$ collide is only $1/m$, that is,
\[\forall x,y\in U,\,x\neq y\colon\ \underset{h\in\mathcal{H}_\text{all}}{\P}\!\left(h(x)=h(y)\right)\,=\frac1m.\]
On the downside, this random process can result in computationally complex hash functions whose evaluation has space and time requirements in $\Theta(u\ln m)$. 
Efficiently computable functions are necessary in order to make applications feasible, however. 
Consequently, the assumption of such a simple uniform hashing remains in large part a theoretical one with little bearing on practical applications; it is invoked primarily to simplify theoretical manipulations. 

The astonishing discovery of small hashing families that can take on the role of $\mathcal{H}_\text{all}$ addresses this problem. 
In their seminal paper in 1979, Larry Carter and Mark Wegman~\cite{carter} introduced the concept of \emph{universal hashing families} and showed that there exist small universal hashing families.
A family $\mathcal{H}$ of hash functions is called \emph{universal} if it can take the place of $\mathcal{H}_\text{all}$ without increasing the collision risk:
\[\mathcal{H}\text{ is universal}\ \iff\ \forall x,y\in U,\,x\neq y\colon\ \underset{h\in\mathcal{H}}{\P}(h(x)=h(y))\le 1/m.\]
In the following years, research has been successfully dedicated to revealing universal hashing families of comparably small size that exhibit the desired properties.

\subsection{Determinism versus Randomization}
Deterministic algorithms cannot keep up with this incredible performance of randomized algorithms,
\ie, as soon as we are forced to choose a single hash function without knowing $S$, there is always a set $S$ such that all keys are mapped to the same cell.
Consequently, deterministic algorithms have not been of much interest.
Using the framework of advice complexity, we measure how much additional information deterministic algorithms need in order to hold their ground or to even have the edge over randomized algorithms.

In order to analyze the advice complexity of hashing, we have to view hashing as an online problem.
This is not too difficult; being forced to take irrevocable decisions without knowing the whole input is the essence of online algorithms, as we describe in the main introduction of this thesis.
In the standard hashing setting, we are required to predetermine our complete strategy, \ie our hash function $h$, with no knowledge about $S$ at all.
In this sense, standard hashing is an ultimate online problem.
We could relax this condition and require that $S=\{s_1,\ldots,s_n\}$ is given piecemeal and the algorithm has to decide to which cell the input $s_i$ is mapped only upon its arrival.
However, this way, an online algorithm could just distribute the input perfectly, which would lead, as in the case of $\mathcal{H}_{all}$, to computationally complex functions.

Therefore, instead of relaxing this condition, we look for the reason why there is such a gap between deterministic and randomized strategies. 
The reason is simple: Deterministic strategies are measured according to their performance in the worst case, whereas randomized strategies are measured according to their performance on average.
However, if we measure the quality of randomized strategies from a worst-case perspective, the situation changes.
In particular, while universal hashing schemes prove incredibly useful in everyday applications, they are far from optimal from a worst-case perspective, as we illustrate now. 

We regard a hash function $h$ as an algorithm operating on the input set of keys $S$.
We assess the performance of $h$ on $S$ by the resulting \emph{maximum cell load}: $\cost(h,S)\coloneqq\alpha_\textnormal{max}\coloneqq\max\{\alpha_1,\ldots,\alpha_m\}$, where $\alpha_i$ is the \emph{cell load} of cell $i$, that is, $\alpha_i\coloneqq \left|h^{-1}(i)\cap S\right|$, the number of keys hashed to cell $i$. This is arguably the most natural cost measurement for a worst-case analysis. Another possibility, the total number of collisions, is closely related.\footnote{This number can be expressed as $\sum_{k=1}^m\binom{\alpha_k}{2}\in\Theta\left(\alpha_1^2+\ldots+\alpha_m^2\right)\subseteq\mathcal{O}\left(m\cdot\alpha_\textnormal{max}^2\right)$.}

The \emph{average load} is no useful measurement option as it is always $\alpha \coloneqq n/m$.
The worst-case cost $n$ occurs if all $n=|S|$ keys are assigned to a single cell. The optimal cost, on the other hand, is $\left\lceil n/m\right\rceil = \left\lceil\alpha\right\rceil$ and it is achieved by distributing the keys of $S$ into the $m$ cells as evenly as possible.

Consider now a randomized algorithm \algA that picks a hash function $h\in \mathcal H_\text{all}$ uniformly at random. A long-standing result \cite{gonnet81}, proven nicely by Raab and Steger~\cite{ballsintobins}, shows the expected cost $\mathbb{E}_{h\in\mathcal{H}_\text{all}}[\cost(h,S)]$ to be in
$\Omega\left(\frac{\ln m}{\ln\ln m}\right)$,
as opposed to the optimum $\lceil\alpha\rceil$. 
The same holds true for smaller universal hashing schemes such as polynomial hashing, where we randomly choose a hash function $h$ out of the family of all polynomials of degree $\mathcal{O}\left(\frac{\ln n}{\ln\ln n}\right)$.
However, no matter which universal hash function family is chosen, it is impossible to rule out the absolute worst-case: For every chosen hash function $h$ there is, due to $u\ge n^2$, a set $S$ of $n$ keys that are all mapped to the same cell.

\subsection{Our Model and the Connection to Advice Complexity}\label{advicecomplexity}
We are, therefore, interested in an alternative hashing model that allows for meaningful algorithms better adapted to the worst case: What if we did not choose a hash function at random but could start with a family of hash functions and then always use the best function, for any set $S$? What is the trade-off between the size of such a family and the upper bounds on the maximum load they can guarantee? In particular, how large is a family that always guarantees an almost optimal load?

In our model, the task is to provide a set $\mathcal{H}$ of hash functions. This set should be small in size while also minimizing the term
\[\cost(\mathcal H,\mathcal S)\coloneqq\max_{S\in\mathcal{S}}\min_{h\in\mathcal{H}}\cost(h,S),\]
which is the best cost bound that we can ensure across a given set $\mathcal{S}$ of inputs by using any hash function in $\mathcal{H}$.
Hence, we look for a family $\mathcal{H}$ of hash functions that minimizes both $|\mathcal{H}|$ and $\cost(\mathcal H,\mathcal S)$. These two goals conflict, resulting in a trade-off, which we parameterize by $\cost(\mathcal H,\mathcal S)$, using the notion of \emph{$c$-ideality}.

\begin{definition}[$\bm c$-Ideality]\label{def:c-ideal}
	Let $c\ge 1$. A function $h:U\to T$ is called \emph{$c$-ideal for a subset $S\subseteq U$} of keys if $\cost(h,S)=\alpha_\textnormal{max}\le c\alpha$.
	In other words, a $c$-ideal hash function $h$ assigns at most $c\alpha$ elements of $S$ to each cell of the hash table $T$.
	
	Similarly, a family of hash functions $\mathcal H$ is called \emph{$c$-ideal for a family $\mathcal S$ of subsets of $U$} if, for every $S \in \mathcal S$, there is a function $h \in \mathcal H$ such that $h$ is $c$-ideal for $S$. This is equivalent to
	$\cost(\mathcal H,\mathcal S) \le c\alpha$.
	If $\mathcal H$ is $c$-ideal for $\mathcal S_n \coloneqq \{S \subseteq U;\ |S| = n\}$ (that is, all sets of $n$ keys), we simply call $\mathcal H$ $c$-ideal.
\end{definition}

We see that $c$-ideal families of hash functions constitute algorithms that guarantee an upper bound on $\cost(\mathcal H)\coloneqq\cost(\mathcal H,\mathcal S_n)$, which fixes one of the two trade-off parameters.
Now, we try to determine the other one and find, for every $c\ge1$, the minimum size of a $c$-ideal family, which we denote by 
\[H_c\coloneqq
\min\{|\mathcal{H}|;\ \cost(\mathcal H) \le c\alpha
\}.\]
Note that $c\ge m$ renders the condition of $c$-ideality void since $n=m\alpha$ is already the worst-case cost; we always have $\cost(\mathcal H)\le n \le c\alpha $ for $c\ge m$. Consequently, every function is $m$-ideal, every non-empty family of functions is $m$-ideal, and $H_m=1$.

With the notion of $c$-ideality, we can now talk about $c$-competitiveness of hashing algorithms.
Competitiveness is directly linked to $c$-ideality since, for each $S$, $\cost(\textsc{Alg}(S))\coloneqq \alpha_\textnormal{max}$ and the cost of an optimal solution is always $\cost(\textsc{Opt}(S))=\lceil \alpha \rceil$.
Therefore, a hash function that is $c$-ideal is $c$-competitive as well.
Clearly, no single hash function is $c$-ideal; however, a hash function family $\mathcal{H}$ can be $c$-ideal and an algorithm that is allowed to choose, for each $S$, the best hash function among a $c$-ideal family $\mathcal{H}$ is thus $c$-competitive.
Such a choice corresponds exactly to an online algorithm with advice that reads $\lceil\log(|\mathcal{H}|)\rceil$~advice bits, which is just enough to indicate the best function from a family of size~$|\mathcal{H}|$.
Since, by definition, there is a $c$-ideal hash function family of size~$H_c$, there is an online algorithm with advice complexity $\lceil\log(H_c)\rceil$ as well.
Moreover, there can be no $c$-competitive online algorithm \algA that reads less than $\lceil\log H_c\rceil$ advice bits: If there were such an algorithm, there would be a $c$-ideal hash function family of smaller size than $H_c$, which contradicts the definition of $H_c$.

We will discuss the relation between $c$-ideal hashing and
advice complexity in more depth in \cref{sec:adv_comp}.

\subsection{Organization}
This paper is organized as follows. In~\cref{sec:contribution}, we give an overview of related work and our contribution.
We present our general method of deriving bounds on the size of $c$-ideal families of hash functions in~\cref{sec:c_ideal}. 
\Cref{sec:estimationofalphamax} is dedicated to precisely calculating these general bounds. In~\cref{sec:edgecases}, we give improved bounds for two edge cases. 
In~\cref{sec:adv_comp}, we analyze the advice complexity of hashing. 
We recapitulate and compare our results in~\cref{hash:sec:conclusion}.

\section{Related Work and Contribution}\label{sec:contribution}
There is a vast body of literature on hashing.
Indeed, hashing still is a very active research area today and we cannot even touch upon the many aspects that have been considered.
Instead, in this section, we focus on literature very closely connected to $c$-ideal hashing.
For a coarse overview of hashing in general, we refer to the survey by Chi and Zhu~\cite{chi_survey} and the seminar report by Dietzfelbinger et al.~\cite{dietzfelbinger_Kurzuebersicht}.

The advice complexity of hashing has not yet been analyzed; however, $c$-ideality is a generalization of perfect $k$-hashing, sometimes also called $k$-perfect hashing.
For $n\le m$ (i.e., $\alpha\le1$) and $c=1$, our definition of $c$-ideality allows no collisions and thus reduces to perfect hashing, a notion formally introduced by Mehlhorn in 1984~\cite{mehlhorn}. 
For this case, Fredman and Komlós~\cite{fredman} proved in 1984 the bounds
\[H_1\in\Omega\left( \frac{m^{n-1}\log(u) (m-n+1)!} {m! \log(m-n+2)} \right) \text{ and } H_1\in\mathcal{O}\left(\frac{-n\log(u)}{\log\left(1-\frac{m!}{(m-n)!m^n}\right)}\right).\] 

In 2000, Blackburn~\cite{Blackburn} improved their lower bound for $u$ large compared to $n$.
Recently, in 2022, Guruswami and Riazonov~\cite{guruswami_perfect_hashing} improved their bound as well.
None of these proofs generalize to $n>m$, that is, $\alpha > 1$.

Another notion that is similar to $c$-ideality emerged in 1995.
Naor et al.~\cite{splitters} introduced $(u,n,m)$-splitters, which coincide with the notion of $1$-ideal families for $\alpha\ge 1$. They proved a result that translates to
\begin{equation}
H_1\in\mathcal{O}\left(\sqrt{2\pi\alpha}^m \euler^\frac{m}{12\alpha}\sqrt{n}\ln u\right). \label{Naorupperbound}
\end{equation}
Since the requirements for $c$-ideality are strongest for $c=1$, this upper bound holds true for the general case of $c\ge1$ as well.

We extend these three results to the general case of $n\ge m$ and $c\ge 1$. Moreover, we tighten the third result further for large $c$.
Specifically, we prove the following new bounds:
\begin{align}
H_c&\ge \lowerboundHc \label{equation_mainlowerbound} 
&&\text{(\Cref{cor:lowerbound_H_c})}\\
H_c&\ge \frac{\ln u - \ln{(c\alpha)}}{\ln m} \label{equation_easylowerbound} &&\text{(\Cref{lowerboundthm1})}\\
H_c&\le \upperboundHc \label{equation_mainupperbound} 
&&\text{(\Cref{cor:upperbound_H_c})}\\
H_c&\in\mathcal{O}\left(\frac{n\ln u}{\ln t}\right)\text{ for any $t\ge1$ and }c\in\omega\left(t\frac{\ln n}{\ln\ln n}\right) \label{equation_Yaobound}
&&\text{(\Cref{upperbound_yao})}
\end{align}
Note that (\ref{equation_mainupperbound}) coincides with (\ref{Naorupperbound}) for $c=1$ and is only an improvement for $\alpha$ and $c$ slightly larger than $1$. Since $H_c \le H_{c'}$ for $c\ge c'$, the bound still improves slightly upon (\ref{Naorupperbound}) in general, depending on the constants hidden in the O-notation in (\ref{Naorupperbound}). Interestingly, the size $u$ of the universe does not appear in (\ref{equation_mainlowerbound}); a phenomenon discussed in~\cite{fredman}. Fredman and Komlós used information-theoretic results based on the Hansel Lemma~\cite{Hansel} to obtain a bound that takes the universe size into account. Körner~\cite{Koerner} expressed their approach in the language of graph entropy.
It is unclear, however, how these methods could be generalized to the case $\alpha>1$ in any meaningful way. 

The straightforward approach of proving a lower bound on $H_1$ is to use good Stirling estimates for the factorials in $H_1 \ge \binom u n / \binom {u/m} \alpha^m$; see Mehlhorn~\cite{mehlhorn}. This yields a lower bound of roughly $\sqrt{2\pi\alpha}^{m-1}/\sqrt{m}$, which is better than (\ref{equation_mainlowerbound}) for $c=1$, see \cref{lem:straightforwardestimation_H_c} for details. Unfortunately, it turns out that we cannot obtain satisfying results for $c>1$ with this approach.
However, the results from Dubhashi and Ranjan~\cite{negativedependence} and the method of Poissonization enable us to circumnavigate this obstacle and derive \cref{equation_mainlowerbound}.

We use our bounds to derive bounds for the advice complexity of hashing, which we will discuss in \cref{sec:adv_comp}.

\section{General Bounds on $H_c$}\label{sec:c_ideal}
We present our bounds on $H_c$, which is the minimum size of $c$-ideal families of hash functions. First, we establish a general lower and upper bound.

We use a \emph{volume bound} to lower-bound $H_c$. We need the following definition. Let $M_c$ be the maximum number of sets $S \in \Sn$ that a single hash function can map $c$-ideally, \ie
\[M_c\coloneqq
\max_{h\in\mathcal{H}_\text{all}}|\{S\in\mathcal{S}_n;\ \alpha_\textnormal{max}\le c\alpha\}| = \max_{h\in\mathcal{H}_\text{all}}|\{S\in\mathcal{S}_n;\  \forall i \in [m]: \alpha_i\le c\alpha\}|.\]

\begin{lemma}[Volume Bound]\label{thm:volumebound}
	The number of hash functions in a family of hash functions that is $c$-ideal is at least the number of sets in $\mathcal{S}_n$ divided by the number of sets for which a single hash function can be $c$-ideal, \ie \[H_c\ge |\mathcal{S}_n| /M_c.\]
\end{lemma}
\begin{proof}
	$|\mathcal{S}_n|=\binom{u}{n}$ is the number of subsets of size $n$ in the universe $U$ and $M_c$ is the maximum number of such subsets for which a single function $h$ can be $c$-ideal. A function family $\mathcal{H}$ is thus $c$-ideal for at most $|\mathcal{H}|\cdot M_c$ of these subsets. If $\mathcal{H}$ is supposed to be $c$-ideal for all subsets---that is, to contain a $c$-ideal hash function for every single one of them---
	we need $|\mathcal{H}|\cdot M_c\ge |\mathcal{S}_n|$.
\end{proof}

To be able to estimate $M_c$, we consider hash functions that distribute the $u$ keys of our universe as evenly as possible:
\begin{definition}[Balanced Hash Function] \label{def:almostequidistributing}
	A \emph{balanced hash function} $h$ partitions the universe into $m$ parts by allotting to each cell $\lceil u/m\rceil$ or $\lfloor u/m\rfloor$ elements.
	We denote the set of balanced hash functions by $\mathcal{H}_\mathrm{eq}$. We write $h_\mathrm{eq}$ to indicate that $h$ is balanced.\footnote{Note that, for any $h_\mathrm{eq}$, there are exactly ($u\!\!\mod m$) cells of size $\lceil {u/m} \rceil$ since $\sum_{k=1}^m |h_\mathrm{eq}^{-1}(k)| = u$.}
\end{definition}

\Cref{thm:equidistribution} states that exactly all balanced hash functions attain the value $M_c$. Therefore, we can limit ourselves to such functions.

\begin{theorem}[Balance Maximizes $M_c$]\label{thm:equidistribution}
	A function is $c$-ideal for the maximal number of subsets if and only if it is balanced.
	In other words, the number of subsets that are hashed $c$-ideally by a hash function equals $M_c$ if and only if $h$ is balanced.
\end{theorem}
\begin{proof}
	For every hash function, the number of keys that are allotted---that is, potentially hashed---to cell $k$ is $\beta_k\coloneqq |h^{-1}(k)|$. This gives us a decomposition $(\beta_1,\ldots,\beta_m)$ of our universe $U$.
	Note that $\beta_1+\ldots+\beta_m=u$ and remember that $h$ is by definition balanced if and only if $\beta_1,\ldots,\beta_m\in\{\lfloor\beta\rfloor,\lceil\beta\rceil\}$, where we use $\beta \coloneqq u/m$. Note that this is in turn equivalent to the condition that $\beta_1,\ldots,\beta_m$ differ pairwise by at most 1.\\
	Let $M_c(\beta_1,\ldots,\beta_m)$ be the number of subsets $S \subseteq U$ for which a hash function $h$ with a decomposition $(\beta_1,\ldots,\beta_m)$ is $c$-ideal.
	Assume that $h$ is not balanced.
	Then, there are $i,j$ such that $\beta_i\le \beta_j+2$.
	Assume without loss of generality that $\beta_1\le \beta_2+2$.
	Choose any element $a$ from $\beta_2$.
	Let $h'$ be the function with $h'(s)=h(s)$ for $s\neq u$ and $h'(a)=1$, that is, $h'$ behaves like $h$ except for mapping $a$ to cell $1$ instead of cell $2$. This function has the decomposition $(\beta_1+1, \beta_2-1, \beta_3,\ldots,\beta_m)$ and is thus more balanced than $h$.
	We show that $h'$ is $c$-ideal for more subsets $S \subseteq U$ than $h$.
	Thus, we show that by moving towards a more balanced distribution, more sets can be mapped $c$-ideally.
	
	Consider any set $S_1 \subset U$ of $n-k$ elements from $U$ with $k\le 2c\alpha$ that does not contain any element from $h^{-1}(1)$ or $h^{-1}(2)$.
	Let $r$ and $r'$ be the number of sets $S_2 \subset h^{-1}(1) \cup h^{-1}(2)$ of size $k$ such that $S=S_1 \cup S_2$ is mapped $c$-ideally by $h$ and $h'$, respectively. We analyze $r'-r$, that is, by how much the number of $c$-ideally mapped sets increases when $a$ is mapped to the first cell instead of the second one. On the one hand, 
	any set $S_2$ containing $a$ and exactly $c\alpha$ elements (different from $a$) from $h^{-1}(2)$ is mapped $c$-ideally by $h'$ but not by $h$. 
	There are $\binom{\beta_2-1}{c\alpha}\binom{\beta_1}{k-c\alpha-1}$ such sets. 
	On the other hand, for any set $S_2$ containing $a$ and exactly $c\alpha$ elements from $h^{-1}(1)$ is not mapped $c$-ideally by $h'$ although it was by $h$. There are $\binom{\beta_1}{c\alpha}\binom{\beta_2-1}{k-c\alpha-1}$ such sets. 
	Any set not containing $a$ is of course hashed identically by $h$ and $h'$. 
	Hence, $r'-r$ is exactly
	\[\binom{\beta_2-1}{c\alpha}\binom{\beta_1}{k-c\alpha-1} - \binom{\beta_1}{c\alpha}\binom{\beta_2-1}{k-c\alpha-1}.\]
	Writing $s=\beta_2-\beta_1-1$ and $x=c\alpha$ and $y=k-c\alpha-1$, we find that this expression is larger than or equal to $0$ if and only if any of the following equivalent conditions holds.
	\begin{align*}
	\binom{\beta_2-1}{c\alpha}\binom{\beta_1}{k-c\alpha-1} &\ge \binom{\beta_1}{c\alpha}\binom{\beta_2-1}{k-c\alpha-1} \\
	\frac{(\beta_1+s)!}{x!(\beta_1+s-x)!}\frac{\beta_1!}{y!(\beta_1-y)!} &\ge
	\frac{\beta_1!}{x!(\beta_1-x)!}\frac{(\beta_1+s)!}{y!(\beta_1+s-y)!} \\
	\frac{(\beta_1+s-y)!}{(\beta_1-y)!} &\ge \frac{(\beta_1+s-x)!}{(\beta_1-x)!} \\
	(\beta_1+s-y)\ldots(\beta_1+s-y-s+1) &\ge (\beta_1+s-x)\ldots(\beta_1+s-x-s+1) \\
	x &\ge y \\
	c\alpha &\ge k-c\alpha-1 \ge 2c\alpha-c\alpha-1 \ge c\alpha-1.
	\end{align*}
	Therefore, whenever a hash function is not balanced, there is another hash function that hashes more sets $c$-ideally.		
\end{proof}

Now fix a balanced hash function $h \in \mathcal{H}_\mathrm{eq}$. 
We switch to a randomization perspective: Draw an $S \in\mathcal{S}_n$ uniformly at random. The cell loads $\alpha_k$ are considered random variables that assume integer values based on the outcome $S\in\mathcal{S}_n$.
The probability that our fixed $h$ hashes the random $S$ $c$-ideally is exactly
\[\mathbb{P}(\alpha_\textnormal{max} \le c\alpha) = \frac{|\{S\in \Sn ;\  \text{$h$ hashes $S$ $c$-ideally}\}|}{|\Sn|}=\frac{M_c}{|\mathcal{S}_n|}. \]

We suspend our analysis of the lower bound for the moment and switch to the upper bound to facilitate the comparison of the bounds.
\begin{lemma}[Probability Bound]\label{thm:probabilitybound2}
	We can bound the minimal size of a $c$-ideal family by
	\[H_c\le\left\lceil\frac{|\mathcal{S}_n|}{M_c}n\ln u\right\rceil.\]
\end{lemma}
\begin{proof}		
	Let $h$ be a hash function and $S \in \Sn$ be a set of $n$ keys.
	Denote by $G_{h,S}$ the characteristic function indicating whether $h$ is $c$-ideal for $S$:
	\[G_{h,S}\coloneqq\begin{cases}1,&\text{if }\alpha_\textnormal{max}\le c\alpha\,,\\0&\text{otherwise}.\end{cases}\]
	Let $B_{h,S} \coloneqq 1-G_{h,S}$ denote the characteristic function indicating whether $h$ is \emph{not} $c$-ideal for $S$. For a subset $S\in\mathcal{S}_n$ chosen uniformly at random, we can thus write the probability that $h$ is $c$-ideal for $S$ as follows:
	\begin{equation}
		\underset{S\in\mathcal{S}_n}{\P}(\alpha_\textnormal{max} \le c\alpha)=1-\underset{S\in\mathcal{S}_n}{\P}(\alpha_\textnormal{max} > c\alpha) = 1-\underset{S\in\mathcal{S}_n}{\mathbb{E}}[B_{h,S}]. \label{eq:upper_bound_hashing_1}
	\end{equation}
	
	Let $h\in\mathcal{H}_\mathrm{eq}$ be any balanced hash function. Since $S$ is chosen uniformly at random, the probability that $h$ is $c$-ideal for $S$ is the number of sets for which $h$ is $c$-ideal divided by the number of sets in total:
	\begin{equation}
		\underset{S\in\mathcal{S}_n}{\P}(\alpha_\textnormal{max} \le c\alpha) = \frac{M_c}{|\mathcal{S}_n|} \eqqcolon p. \label{eq:upper_bound_hashing_2}
	\end{equation}
		
	We now turn towards families of hash functions.
	For a family of hash functions $\mathcal{H}$, the product $B_{\mathcal{H},S}\coloneqq\prod_{h\in\mathcal{H}}B_{h,S}$
	is equal to 1 if and only if no $h\in\mathcal{H}$ is $c$-ideal for $S$.
	Let $S \in \Sn$ be again any set of $n$ keys. Let $\mathcal{H}\subseteq\mathcal{H}_\mathrm{eq}$ be a family of hash functions that consists of an arbitrary number of equi\-distributing hash functions chosen uniformly at random with replacement.
	The probability that this $\mathcal{H}$ is not $c$-ideal for $S$ is
	\begin{align*}
		\P(\mathcal{H} \text{ is not $c$-ideal for $S$}) &=\mathbb{E}[B_{\mathcal{H},S}]\\
		\cmt{$h$ chosen independently}&=\mathbb{E}[\prod_{h\in\mathcal{H}}\left(B_{h,S}\right)]\\
		\cmt{linearity of expectation}	&=\prod_{h\in\mathcal{H}}\mathbb{E}[B_{h,S}]\\
		\cmt{Equation (\ref{eq:upper_bound_hashing_1}) and (\ref{eq:upper_bound_hashing_2})}&=\prod_{h\in\mathcal{H}}\left(1-p\right)\\
		\cmt{independent of hash function}&=\left(1-p\right)^{|\mathcal{H}|}.
	\end{align*}
	The union bound yields
	$|\mathcal{S}_n|\left(1-p\right)^{|\mathcal{H}|}$ as an upper bound for the probability that $\mathcal{H}$ is not $c$-ideal for $\Sn$:
	\begin{align*}
		\P(\mathcal{H} \text{ is not $c$-ideal for $\Sn$}) &= \P\left(\bigcup_{S \in \Sn}\mathcal{H} \text{ is not $c$-ideal for $S$}\right) \\
		\cmt{union bound}&\le \sum_{S \in \Sn} \P(\mathcal{H} \text{ is not $c$-ideal for $S$}) \\
		&= |\Sn|\left(1-p\right)^{|\mathcal{H}|}.
	\end{align*}
	If this probability is less than 1, then the converse probability, $\mathbb{P}(\alpha_\textnormal{max}\le c\alpha)$, 
	is larger than 0 and we can infer that at least one $c$-ideal family of size $|\mathcal{H}|$ exists. Consider an $\mathcal{H}$ of size $|\mathcal{H}|>-\ln |\mathcal{S}_n|/\ln (1-p)$:
	\begin{align*}
		|\mathcal{H}|&>\frac{-\ln |\mathcal{S}_n|}{\ln (1-p)} \iff \\
		\ln \left((1-p)^{|\mathcal{H}|}\right) &>-\ln |\mathcal{S}_n| \Longrightarrow \cmt{exponentiate}\\
		(1-p)^{|\mathcal{H}|} &> \frac{1}{|\Sn|} \Longleftrightarrow\\
		|\mathcal{S}_n|\left(1-p\right)^{|\mathcal{H}|}&<1 \Longrightarrow \\ \mathbb{P}(\alpha_\textnormal{max}> c \alpha) &<1.
	\end{align*}
	Thus, if $\mathcal{H}$ is strictly larger than the chosen value, it is $c$-ideal for $\mathcal{S}_n$.
	The minimal number of hash functions needed to obtain a $c$-ideal family of hash functions is thus at most $1$ plus the chosen value:
	\[H_c\le1+\left\lfloor\frac{-\ln |\mathcal{S}_n|}{\ln(1-p)}\right\rfloor.\]
	We simplify this bound with the first-order Taylor estimate for the natural logarithm, which has the expansion $\ln(1+x)=x-x^2/2+x^3/3-x^4/4+\ldots$ around 1, and use $|\mathcal{S}_n|=\binom un\le u^n$ to conclude the proof:
	\[
	H_c\le1+\left\lfloor\frac{\ln |\mathcal{S}_n|}{-\ln(1-p)}\right\rfloor \le\left\lceil\frac{\ln |\mathcal{S}_n|}{p}\right\rceil
	\le\left\lceil\frac {|\mathcal{S}_n|}{M_c}\ln |\mathcal{S}_n| \right\rceil\le\left\lceil\frac {|\mathcal{S}_n|}{M_c}n\ln u\right\rceil.
	\]
	
\end{proof}

We combine \cref{thm:volumebound} and \cref{thm:probabilitybound2} and summarize our findings:
\begin{corollary}[General Bounds on Family Size $\bm{H_c}$]\label{cor:mainbounds}
	The size of a $c$-ideal family of hash functions is bounded by
	\[
	\frac 1 {\P(\alpha_\textnormal{max}\le c\alpha)} \le H_c \le \frac {n \ln{u}} {\P(\alpha_\textnormal{max}\le c\alpha)}.\]
\end{corollary}

Now, to lower-bound $H_c$, we first consider for $c=1$ a straightforward application of \cref{thm:volumebound} suggested by Mehlhorn~\cite{mehlhorn} and then ponder whether we could extend this approach for $c$ larger than $1$.
\begin{lemma}\label{lem:straightforwardestimation_H_c}
	The number of hash functions in a $1$-ideal family of hash functions is bounded from below by approximately
	\[\approx \frac{\sqrt{2\pi\alpha}^{m-1}}{\sqrt{m}}.\]
\end{lemma}
\begin{proof}
	$M_1$, the number of sets that a single hash function can map $1$-ideal is attained by balanced hash functions.
	A balanced hash function divides the universe into $m$ equal parts of size $u/m$. Exactly the sets that consist of $\alpha = n/m$ keys from each of the $m$ parts of the universe are mapped $1$-ideally.
	There are exactly 
	\begin{equation}
		M_1 = \binom{u/m}{\alpha}^m. \label{eq:M1}
	\end{equation}
such sets. We write again $\beta \coloneqq u/m$ and use the volume bound to estimate $H_1$ as follows:
	\begin{align*}
		\cmt{\cref{thm:volumebound}} H_1{}&\ge\frac {\mathcal{S}_n}{M_1}\\
		\cmt{Equation (\ref{eq:M1})}&=\frac{\binom un\hphantom{{}^m}}{{\binom\beta\alpha}^m}\\
		\cmt{Definition of binomial}&=\frac{u!(\beta-\alpha)!^m\alpha!^m}{(u-n)!n!\beta!^m}
	\end{align*}
	These factorials can be estimated with Stirling's formula. We are going to use the following estimates by Robbins~\cite{stirling}:
	\begin{equation*}
		\sqrt{2\pi n} \left(\frac{n}{e}\right)^n e^{\frac1 {12n+1}}\le n! \le \sqrt{2\pi n} \left(\frac{n}{e}\right)^n e^{\frac1 {12n}}.
	\end{equation*}
	We insert these estimations for the factorials and then reorder the terms repeatedly to simplify the expression:
	\begin{align*}
		&\ge\frac{\sqrt{2\pi u}\sqrt{2\pi (\beta-\alpha)}^m\sqrt{2\pi \alpha}^m}{\sqrt{2\pi (u-n)}\sqrt{2\pi n}\sqrt{2\pi \beta}^m}\cdot\frac{\left(\frac {u}e\right)^{u}\left(\frac{\beta-\alpha}e\right)^{m(\beta-\alpha)}\left(\frac {\alpha}e\right)^{m\alpha}}{\left(\frac {u-n}e\right)^{(u-n)}\left(\frac {n}e\right)^{n}\left(\frac {\beta}e\right)^{m\beta}}\cdot\\
		&\hphantom{{}={}}\cdot\frac{e^{\frac1{12u+1}}e^{\frac1{12(\beta-\alpha)+1}}e^{\frac1{12\alpha+1}}}{e^{\frac1{12(u-n)}}e^{\frac1{12n}}e^{\frac1{12\beta}}}\\
		&=\sqrt{2\pi}^{2m+1-(2+m)}\sqrt{\frac{u(\beta-\alpha)^m\alpha^m}{(u-n)n\beta^m}}\cdot
		\frac{\left(\frac {u}e\right)^{u}\left(\frac{\beta-\alpha}e\right)^{u-n}\left(\frac {\alpha}e\right)^{n}}{\left(\frac {u-n}e\right)^{u-n}\left(\frac {n}e\right)^{n}\left(\frac {\beta}e\right)^{u}}\cdot \\
		&\hphantom{{}={}}\cdot
		e^{\frac1{12u+1}+\frac1{12(\beta-\alpha)+1}+\frac1{12\alpha+1}-\frac1{12(u-n)}-\frac1{12n}-\frac1{12\beta}}.
		\end{align*}
	We pause for a moment and abbreviate $A\coloneqq e^{\frac1{12u+1}+\frac1{12(\beta-\alpha)+1}+\frac1{12\alpha+1}-\frac1{12(u-n)}-\frac1{12n}-\frac1{12\beta}}$ before we continue our calculation:
	\begin{align*}
		\cmt{$\beta = u/m$}&=\sqrt{2\pi}^{m-1}\sqrt{\frac{u(u-n)^m\alpha^m}{(u-n)nu^m}}\cdot
		\frac{u^{u}\left(\beta-\alpha\right)^{u-n}\alpha^{n}}{\left(u-n\right)^{u-n}n^{n}\beta^{u}}\cdot A\\
		&=\sqrt{2\pi}^{m-1}\sqrt{\frac{\left(\frac{u-n}{u}\right)^{m-1}\alpha^m}{n}}\cdot 
		\frac{{u^{u}}{\left(u-n\right)^{u-n}}{n^{n}}}{{\left(u-n\right)^{u-n}}{n^{n}}{u^{u}}}\cdot A\\
		\cmt{$\frac{{u^{u}}{\left(u-n\right)^{u-n}}{n^{n}}}{{\left(u-n\right)^{u-n}}{n^{n}}{u^{u}}}=1$}&=\sqrt{2\pi\alpha}^{m-1}\sqrt{\left(\frac{u-n}{u}\right)^{m-1}}\sqrt{\frac{\alpha}{n}}\cdot A\\
		\cmt{$\alpha = n/m$}&=\sqrt{2\pi\alpha}^{m-1}\sqrt{\left(1-\frac{n}{u}\right)^{m-1}}\sqrt{\frac1m}\cdot A.
	\end{align*}
	The factor $A$ is smaller than $e$ and tends to $1$ with $\alpha \to \infty$.
Hence, this is approximately
	\begin{align*}
		&\approx\sqrt{2\pi\alpha}^{m-1}\sqrt{\left(1-\frac{n}{u}\right)^{m-1}}\sqrt{\frac1m}\\
		&\approx\frac1{\sqrt m}\sqrt{2\pi\alpha}^{m-1}.
	\end{align*}
	
\end{proof}
The natural extension of this approach yields 
$H_c \ge K(n,m,c\alpha)\binom u n / \binom{u/m}{c\alpha}^m$,
where $K(n,m,c\alpha)$ denotes the number of compositions of $n$ into $m$ non-negative integers between $0$ and $c\alpha$. 
This factor,
which accounts for the number of possibilities to split the $n$ keys into $m$ different cells,
equals 1 for $c=1$.
For general $c$, however, to the best of our knowledge, even the strongest approximations for $K(n,m,c\alpha)$ do not yield a meaningful lower bound.
Therefore, we are forced to use a different strategy for general $c$ and we estimate $M_c$ via the probability $\P(\alpha_\textnormal{max}\le c\alpha)$ in the following section.

\section{Estimations for $\mathbb{P}(\alpha_\textnormal{max}\le c\alpha)$ }\label{sec:estimationofalphamax}
It remains to find good bounds on the probability $\mathop{\mathbb{P}}(\alpha_\textnormal{max}\le c\alpha
)$, which will immediately yield the desired bounds on $H_c$.
We start by establishing an upper bound on $\mathop{\mathbb{P}}(\alpha_\textnormal{max}\le c\alpha
)$.
\subsection{Upper Bound}
Recall that we fixed a balanced hash function $h \in \mathcal{H}_{eq}$ and draw an $S \in \Sn$ uniformly at random. For every cell $k \in T$, we model its load $\alpha_k$ as a random variable.
The joint probability distribution of the $\alpha_i$ follows the hypergeometric distribution, \ie
\begin{align*}
\P((\alpha_1,\ldots,\alpha_m)=(\ell_1,\ldots,\ell_m))&=
\frac{\binom{|h^{-1}(1)|}{\ell_1}\cdots\binom{|h^{-1}(m)|}{\ell_m}}{\binom un}.
\end{align*}
There are two obstacles to overcome. First, calculating with probabilities without replacement is difficult.
Second, the sum of the $\alpha_i$ is required to be $n$; in particular, the variables are not independent.
The first obstacle can be overcome by considering drawing elements from $U$ with replacement instead, \ie drawing multisets instead of sets. As the following lemma shows, we do not lose much by this assumption.
\begin{lemma}[Switching to Drawing With Replacement] \label{thm:replacement}
	For any fixed $\varepsilon>0$, there are $u$ and $n$ large enough such that
	\[\frac{M_c}{|\Sn|} \le (1+\varepsilon)\P(T_i \le c\alpha, i \in [m]), \]
	where we use $T_1, \ldots T_m$ to model the cell loads as binomially distributed random variables with parameters $n$ and $1/m$. In other words, for $u$ and $n$ large enough, we can consider drawing the elements of $S \in \Sn$ with replacement and only lose a negligible factor.
\end{lemma}
\begin{proof}
	We use $\mathcal{K}_n$ to denote the set of multisets with elements from $U$ of cardinality $n$.
	There are $\binom{u+n-1}{n}$ such multisets.\footnote{Maybe the best way to see this is to observe the following bijection between multisets of size $n$ with elements from a set of size $u$ and normal sets of size $n$ consisting of elements of a set of size $u+n-1$: Write the elements $x_1,\ldots,x_n$ of the multiset in non-decreasing order and then map this multiset to the set $\{x_1,x_2+1,\ldots,x_n+n-1\}$.}
	We write $K_c$ for the number of those multisets of $\mathcal{K}_n$ that are hashed $c$-ideally by $h$. Note that $\P(T_i \le c\alpha, i \in [m])$ equals $K_c / |\mathcal{K}_n|$.
	Clearly, $M_c$ is less than or equal to $K_c$ since every true set is a multiset. It follows that $M_c/|\Sn| \le K_c / |\Sn|$. To conclude, we multiply the right-hand side with $|\Sn| / |\mathcal{K}_n|$ and observe that
	\[
	\frac{|\mathcal{S}_n|}{|\mathcal{K}_n|} = 
	\frac{\binom{u}{n}}{\binom{u+n-1}{n}} = 
	\prod_{k=0}^{n-1}\frac{(u-k)}{(u+n-1-k)} \; \xrightarrow[n/u \to 0]{u \to \infty} \; 1.\]
\end{proof}

If we consider drawing elements from $U$ with replacement, we have a multinomial distribution, \ie
\begin{align*}
\P((T_1,\ldots,T_m)=(\ell_1,\ldots,\ell_m))&=
\binom{n}{\ell_1,\ldots,\ell_m}/m^n.
\end{align*}
This is easier to handle than the hypergeometric distribution. The $T_i$, $i\in[m]$, are still not independent, however.
The strong methods by Dubhashi and Ranjan~\cite{negativedependence} provide us with a simple way to overcome this second obstacle:

\begin{lemma}[Proposition 28 and Theorem 31 from~\cite{negativedependence}]\label{lem:negativeregression}
	The joint distribution of the random variables $\alpha_i$ for $i\in[m]$ is upper-bounded by their product:
	\[\P(\alpha_i\le c\alpha, i\in [m]) \le \prod_{i=1}^{m}\P(\alpha_i\le c\alpha). \]
	The same holds if we consider the $T_i$, $i\in[m]$, instead of the $\alpha_i$.
\end{lemma}

We are ready to give an upper bound on the probability that a hash function family is $c$-ideal.

\begin{theorem}[Upper Bound on $\P(\alpha_\textnormal{max} \le c\alpha)$]\label{thm:excessprobability}
	For arbitrary $\varepsilon>0$,
	\[\P(\alpha_\textnormal{max} \le c\alpha) \le (1+\varepsilon)\exp\left(-m \euler^{-\alpha}(1-\varepsilon)\left(\frac \alpha {c\alpha+1}\right)^{c\alpha+1}\right),\]
	where $n$ tends to infinity.
\end{theorem}
\begin{proof}
	We can bound the probability that a balanced hash function is $c$-ideal as follows.
	\begin{align*}
	&= \P(\alpha_\textnormal{max} \le c\alpha) \\
	\cmt{Definition of maximum}&= \P(\alpha_i \le c\alpha, i \in [m]) \\
	\cmt{Switching to binomial variables, \cref{thm:replacement}}&\le (1+\varepsilon)\P(T_i \le c\alpha, i \in [m]) \\
	\cmt{\Cref{lem:negativeregression}}&\le (1+\varepsilon)\prod_{i=1}^{m}\P(T_i\le c\alpha)\\
	\cmt{$T_i$, $i\in[m]$, are identically distributed}&= (1+\varepsilon)\P(T_1 \le c\alpha)^m \\
	\cmt{Converse probability}&= (1+\varepsilon)\left(1-\P(T_1 > c\alpha)\right)^m \\
	\cmt{$e^{-x}\ge 1-x$ for $x\in(0,1)$}&\le (1+\varepsilon)\euler^{-m\P(T_1 > c\alpha)}.
	\end{align*}	
	To continue our estimation, we need to bound $\P(T_1 > c\alpha)$ from below.
	We start with the definition of the binomial distribution, take only the first summand and use the standard bound $\binom{n}{k}\ge (n/k)^k$ on the binomial coefficient to obtain the following.
	\begin{align*}
	\P(T_1 > c\alpha)\qquad&= \sum_{k=c\alpha+1}^{n} \binom n k \left(\frac 1 m \right)^k \left(1-\frac 1 m \right)^{n-k} \\
	\cmt{Rearrange}&= \left(1-\frac 1 m\right)^n \sum_{k=c\alpha+1}^{n} \binom n k \left(\frac 1 {m-1} \right)^k \\
	\cmt{$\frac n m = \alpha$ and $\binom n k \ge \left(\frac n k\right)^k$}&\ge \left(1-\frac \alpha n\right)^n \sum_{k=c\alpha+1}^{n} \left(\frac n k\right)^k \left(\frac 1 {m-1} \right)^k \\
	\cmt{Rearrange}&= \left(1-\frac \alpha n\right)^n \sum_{k=c\alpha+1}^{n} \left(\frac 1 k\right)^k \left(\frac n {m-1} \right)^k \\
	\cmt{$\frac n {m-1} = \frac n m + \frac n {m(m-1)} \ge \frac n m = \alpha$}&\ge \left(1-\frac \alpha n\right)^n \sum_{k=c\alpha+1}^{n} \left(\frac 1 k\right)^k \alpha^k \\
	\cmt{Take only $k=c\alpha+1$}&= \left(1-\frac \alpha n\right)^n \left(\frac \alpha {c\alpha+1}\right)^{c\alpha+1}\\
	\cmt{$\left(1-\frac \alpha n\right)^n \rightarrow \euler^{-\alpha}$ for $n$ large enough}&\ge (1-\varepsilon')\euler^{-\alpha} \left(\frac \alpha {c\alpha+1}\right)^{c\alpha+1}.
	\end{align*}
	If we combine the two estimates and choose $\varepsilon'$ appropriately, we achieve the desired result.
\end{proof}
\laterdef{\lowerboundHc}{(1-\varepsilon) \exp\left(\frac m {\euler^{\alpha}}(1-\varepsilon)\left(\frac \alpha {c\alpha+1}\right)^{c\alpha+1}\right)}

\Cref{cor:mainbounds} translates the upper bound from \cref{thm:excessprobability} into a lower bound on $H_c$.
\begin{theorem}[Lower Bound on Family Size $H_c$]\label{cor:lowerbound_H_c}
	For arbitrary $\varepsilon>0$, the number of hash functions in a $c$-ideal family of hash functions is bounded from below by
	\[H_c \ge \frac 1 {\P(\alpha_\textnormal{max} \le c\alpha)} \ge \lowerboundHc.
	\]
\end{theorem}

\subsection{Lower Bound}
Another way to overcome the obstacle that the variables are not independent would have been to apply a customized Poissonization technique. The main monograph presenting this technique is in the book by Barbour et al.~\cite{poissonizationbook}; the textbook by Mitzenmacher and Upfal~\cite{mitzenmacher} gives a good illustrating example.
A series of arguments would have allowed us to bound the precision loss we incur by a constant factor of 2, leading to
\[\P(T_\textnormal{max}\le c\alpha) \le 2 \P(Y \le c\alpha)^m, \]
where we use $T_\textnormal{max}$ to denote $\max\{T_1,\ldots,T_m\}$ and where $Y$ is a Poisson random variable with mean $\alpha$.
By using \cref{lem:negativeregression}, we were able to abbreviate this approach.
However, part of this Poissonization technique can be used for the lower bound on $\P(\alpha_\textnormal{max} \le c\alpha)$.

Recall that a random variable $X$ following the Poisson distribution---commonly written as $X\sim P_\lambda$ and referred to as a Poisson variable---takes on the value $k$ with probability $\mathbb{P}(X=k)=\frac1{\euler^\lambda}\frac{\lambda^k}{k!}$.
The following lemma is the counterpart to \cref{thm:replacement}.
\begin{lemma}\label{thm:replacement2}
	We can bound $M_c/|\mathcal{S}_n|=\P(\alpha_\textnormal{max}\le c\alpha)$ from below by
	$\P(T_\textnormal{max}\le c\alpha)$.
\end{lemma}
\begin{proof}
	Fix an arbitrary $S_0\in\mathcal{S}_n$ and choose an $h\in\mathcal{H}_\text{all}$ uniformly at random. Then each of the $n$ keys in $S_0$, hashed one after another, has the same probability of $1/m$ to be hashed to a given cell $k$.
	Note that this is equivalent to the definition of the $T_k$. Moreover, let $\chi(h,S)$ denote the characteristic function of $c$-ideality, which is $1$ if $\max\limits_{\!k\in\{1,\ldots,m\}\!}\{|h^{-1}(k)\cap S|\}\le c\alpha$ and $0$ otherwise.
	\begin{align*}
	&\hphantom{{}={}}\mathbb{P}(T_\textnormal{max}\le c\alpha)\\
	\cmt{Definitions of $\chi$ and $T_k$}&=\mathop{\mathbb{P}}\limits_{h\in\mathcal{H}_\text{all}}(\chi(h,S_0)=1)\\
	\cmt{Definition of $\mathbb{E}$}&=\mathop{\mathbb{E}}\limits_{h\in\mathcal{H}_\text{all}}[\chi(h,S_0)]\\
	\cmt{Valid for arbitrary $S_0\in\mathcal{S}_n$}&=\mathop{\mathbb{E}}\limits_{S\in\mathcal{S}_n}\left[\mathop{\mathbb{E}}\limits_{h\in\mathcal{H}_\text{all}}[\chi(h,S)]\right]\\
	\cmt{Definition of $\mathbb{E}$}&=\sum_{S\in\mathcal{S}_n}\mathbb{P}(S)\sum_{h\in\mathcal{H}_\text{all}}\mathbb{P}(h)\cdot\chi(h,S)\\
	\cmt{Swap sums}&=\sum_{h\in\mathcal{H}_\text{all}}\mathbb{P}(h)\sum_{S\in\mathcal{S}_n}\mathbb{P}(S)\cdot\chi(h,S)\\
	\cmt{Definition of $\mathbb{E}$}&=\mathop{\mathbb{E}}\limits_{h\in\mathcal{H}_\text{all}}\left[\mathop{\mathbb{E}}\limits_{S\in\mathcal{S}_n}\left[\chi(h,S)\right]\right]
	\end{align*}

	For every $h\in\mathcal{H}_\text{all}$, we have
	$\sum_{S\in\mathcal{S}_n}\chi(h,S)\le M_c$ by the definition of $M_c$, whence
	\[\mathop{\mathbb{E}}\limits_{S\in\mathcal{S}_n}[\chi(h,S)]=\mathop{\mathbb{P}}\limits_{S\in\mathcal{S}_n}(\chi(h,S)=1)=\frac{\sum_{S\in\mathcal{S}_n}\chi(h,S)}{\sum_{S\in\mathcal{S}_n}1\hphantom{(h,S)}}\le\frac{M_c}{|\mathcal{S}_n|}.\]
	In particular, the inequality still holds true for the expected value over all $h\in\mathcal{H}_\text{all}$. Together with the equality derived above we conclude
	\begin{align*}
	\mathbb{P}(T_\textnormal{max}\le c\alpha)=\mathop{\mathbb{E}}\limits_{h\in\mathcal{H}_\text{all}}\left[\mathop{\mathbb{E}}\limits_{S\in\mathcal{S}_n}\left[\chi(h,S)\right]\right]\le\frac{M_c}{|\mathcal{S}_n|}.
	\end{align*}
\end{proof}

For the counterpart to \cref{thm:excessprobability}, we use Poissonization to turn the binomial variables into Poisson variables.

The next lemma shows that $(T_1,\ldots,T_m)$ has the same probability mass function as $(Y_1,\ldots,Y_n)$ under the condition $Y=n$.

\begin{lemma}[Sum Conditioned Poisson Variables]\label{lm:distributionequal}
	For any natural numbers $\ell_1,\ldots,\ell_m$ with $\ell_1+\ldots+\ell_m= n$, we have that 
	\[\mathbb{P}((Y_1,\ldots,Y_m)=(\ell_1,\ldots,\ell_m)\mid Y=n)=\mathbb{P}((T_1,\ldots,T_m)=(\ell_1,\ldots,\ell_m)).\]
\end{lemma}
\begin{proof}
	\begin{align*}
	\mathbb{P}((Y_1,\ldots,Y_m)=(\ell_1,\ldots,\ell_m)\,\mid\,{Y}=n)&=\frac{\mathbb{P}(Y_1=\ell_1)\cdot\ldots\cdot \mathbb{P}(Y_m=\ell_m)}{\mathbb{P}({Y}=n)}\\
	\cmt{$Y_1,\ldots,Y_m\sim P_\alpha$ and ${Y}\sim P_n$}&=\frac{\frac1{\euler^\alpha}\frac{\alpha^{\ell_1}}{\ell_1!}\cdot\ldots\cdot \frac1{\euler^\alpha}\frac{\alpha^{\ell_m}}{\ell_m!}}{\frac1{\euler^n}\frac{n^n}{n!}}\\
	\cmt{Cancel using $\alpha m=n$}&=\frac{\frac{\alpha^{\ell_1+\ldots+\ell_m}}{\ell_1!\cdot\ldots\cdot \ell_m!}}{\frac{n^n}{n!}}\\
	\cmt{Eliminate double fraction}&=\frac{\alpha^{\ell_1+\ldots+\ell_m}\cdot n!}{\ell_1!\cdot\ldots\cdot \ell_m!\cdot n^n}\\
	\cmt{$\ell_1+\ldots+\ell_m=n$ and $\frac n\alpha=m$}&=\frac{n!}{\ell_1!\cdot\ldots\cdot \ell_m!\cdot m^n}\\
	\cmt{Multinomial $\binom{n}{\ell_1,\ldots,\ell_m}\coloneqq\frac{n!}{\ell_1!\cdot\ldots\cdot \ell_m!}$}&=\frac{\binom{n}{\ell_1,\ldots,\ell_m}}{m^n}\\
	\cmt{Definition of multinomial distribution}&=\mathbb{P}((T_1,\ldots,T_m)=(\ell_1,\ldots,\ell_m))\\
	\end{align*}
	
\end{proof}
The result of \cref{lm:distributionequal} immediately carries over to the conditioned expected value for any real function $f(\ell_1,\ldots,\ell_m)$, by the definition of expected values.

\begin{corollary}[Conditioned Expected Value]\label{lm:generalpoisson}
	Let $f(\ell_1,\ldots,\ell_m)$ be any real function. We have
	$\mathbb{E}\left[f(Y_1,\ldots,Y_m)\mid Y=n\right]=\mathbb{E}\left[f(T_1,\ldots,T_m)\right].$
\end{corollary}

We are ready to state the counterpart to \cref{thm:excessprobability}.
\begin{lemma}[Lower Bound on Non-Excess Probability]\label{thm:upperbound}
	Set $d=c\alpha$. We have \[\P(T_\textnormal{max}\le d) \ge \frac{\sqrt{2\pi n}}{(2\pi d)^{m/(2c)}}\frac{1}{c^n}\frac{1}{\euler^{\frac{m}{12cd}}}(\alpha+1)^{m(1-\frac1c)}.\]
\end{lemma}
\begin{proof}
	We only sketch the proof. Using \cref{lm:generalpoisson}, we can formulate this probability with Poisson variables:
	\begin{align*}
	\mathbb{P}(T_\textnormal{max}\le d) &= \mathbb{E}[\chi_{T_\textnormal{max}\le d}]\\
	\cmt{\Cref{lm:generalpoisson}}&= \mathbb{E}[\chi_{\max\{Y_1,\ldots,Y_m\}\le d} \mid Y=n]\\
	\cmt{Definition of conditional probability}&= \frac{\mathbb{E}[\chi_{\{Y=n\}}\chi_{\max\{Y_1,\ldots,Y_m\}\le d}]}{\mathbb{P}(Y=n)}.
	\end{align*}
	We know that $Y \sim P_n$; hence, the denominator equals $n^n / (n! \euler^n)$.
	For the numerator, we use the fact that $Y_i \sim P_\alpha$ for all $i$ to write:
	\begin{align*}
	\mathbb{E}[\chi_{\{Y=n\}}\chi_{\max\{Y_1,\ldots,Y_m\}\le d}] &= \sum_{\substack{
			\ell_1,\ldots,\ell_m\in\N\\
			\ell_1+\ldots+\ell_m = n\\
			\ell_1,\ldots,\ell_m \leq d}} \prod_{i=1}^m \frac{1}{\euler^\alpha} \frac{\alpha^{\ell_i}}{\ell_i!}
	\end{align*}
	We estimate the sum by finding a lower bound for both the number of summands and the products that constitute the summands. The product attains its minimum when the $\ell_i$ are distributed as asymmetrically as possible, that is, if almost all $\ell_i$ are set to either $d$ or $0$, with only one being $(n\mod d)$. This fact is an extension of the simple observation that for any natural number $n$, we have
	\[n!n! = n(n-1)(n-2)\cdot\ldots\cdot 1 \cdot n! < 2n(2n-1)(2n-2)\cdot\ldots\cdot(n+1) \cdot n! = (2n)!.\]
	A rigorous proof is obtained by extending the range of the expression to the real numbers, taking the derivative, and analyzing the extremal values.
	\begin{lemma}\label{lem:technicallysophisticated}
		We have
		\begin{align*}
		\min_\Omega \left\{\prod_{i=1}^m \frac{1}{\ell_i!}\right\}=1/(d!)^{m/c},
		\end{align*}
		where $\Omega\coloneqq\{(\ell_1,\ldots,\ell_m)\in\R^m;\ \ell_1+\ldots+\ell_m=n\wedge 0\le\ell_1,\ldots,\ell_m\le d\}$.
	\end{lemma}
	\begin{proof}
 	Let $f(\ell_1,\ldots,\ell_m)\coloneqq\prod_{i=1}^m \frac{1}{\ell_i!}$. Note that the relaxation used in the definition of $f$ from $\ell_1,\ldots,\ell_m\in\N$ to $0\le\ell_1,\ldots,\ell_m\in\R$ can only increase the minimization range, decreasing the entire term further.
 	
 	First, we use that $f$ is extremal exactly if its logarithm $g\coloneqq \log(1/f)=\sum_{i=1}^m\log \ell_i!$ is extremal.
 	We now imagine that this is a hyperplane in $m$ dimensions, and we want to take the derivative 
 	and set it to zero in order to determine the minimum.
 	The derivative in all dimensions corresponds to the gradient, which is
 	$\nabla g=(\partial_{\ell_1}\log \ell_1!,\ldots,\partial_{\log \ell_m}\log \ell_m!)$.
 	
 	To calculate the derivative of a factorial, we use the Gamma function, which is defined as
 	$\Gamma(z)=\int_0^\infty t^{z-1}\euler^{-t}\,\textrm{d}t$ and which satisfies
 	$\Gamma(z)=(z-1)!$.
 	The components of the gradient are thus the logarithmic derivative of $\Gamma(z)$.
 	The derivative of the Gamma function is
 	\[\Gamma'(z) = \int_0^\infty \partial_z t^{z-1}\euler^{-t}\,\textrm{d}t = \int_0^\infty \partial_z \euler^{(z-1)\ln(t)}\euler^{-t}\,\textrm{d}t = \int_0^\infty \ln(t)t^{z-1}\euler^{-t}\,\textrm{d}t.\]
 	These values are strictly monotonically increasing because the second derivative, which is
 	\[\partial_z^2\log z!=\partial_z\frac{\Gamma'(z)}{\Gamma(z)}=\frac{\Gamma(z)\Gamma''(z)-\Gamma'(z)\Gamma'(z)}{\Gamma(z)^2},\]
 	is strictly positive: To see that the numerator is positive, we use the Cauchy-Schwarz inequality, which states
 	\[\langle v,v\rangle \langle w,w\rangle - \langle v,w\rangle ^2 \ge 0 .\]
 	We can then take the scalar product on functionals
 	\[\langle v,w\rangle(z)\coloneqq\int_0^\infty v(t)w(t)\frac{t^{z-1}}{\euler^t}\,\textrm{d}t\]
 	with $v=1$ and $w=\ln$ and obtain exactly 
 	$\Gamma(z)\Gamma''(z)-\Gamma'(z)\Gamma'(z) = \langle v,v\rangle\langle w,w\rangle-\langle v,w\rangle\langle w,v\rangle \ge 0$.
 	
 	The normal vector $(1,\ldots,1)$ of the hyperplane defined by the side condition $\ell_1+\ldots+\ell_m-n=0$ is collinear with the gradient for $\ell_1=\ldots=\ell_m=\alpha$ and nowhere else by the proven monotonicity.
 	This shows that $f$ has only a single inner extremum with value $1/\alpha^m$ attained at $\ell_1=\ldots=\ell_m=\alpha$, since the projection of the gradient onto the hyperplane---which gives the direction of steepest ascent under the given restrictions---vanishes if and only if it is perpendicular to the hyperplane, that is, collinear with the normal vector.
 	
 	We must also check the boundary for extrema.
 	Recall from the statement of the lemma that our area is
 	$\Omega\coloneqq\{(\ell_1,\ldots,\ell_m)\in\R^m;\ \ell_1+\ldots+\ell_m=n\wedge 0\le\ell_1,\ldots,\ell_m\le d\}$.
 	 All candidates $(\ell_1,\ldots,\ell_m)\in\partial\Omega$---that is, on the boundary of our 
 	 area---have at least one $\ell_i$ equal to either zero or $d$. 
 	 We can now repeat our entire argument with the new area being the boundary of $\Omega$ and obtain that again, the extremal value is exactly if all $\ell_i$ are equal except for the new boundary.
 	 
 	 Iterating the entire argument $m$ times shows that the global boundary extrema are attained whenever $\ell_1,\ldots,\ell_m\in\{0,d\}$. Under this assumption, the restriction $\ell_1+\ldots+\ell_m=n$ implies that $\ell_i=d$ for exactly $n/d=m/c$ of the $m$ indices $i\in\{1,\ldots,m\}$ and $\ell_i=0$ for the rest, yielding the global minimal value $(d!)^{m/c}$.
	\end{proof}

	The next step is to find a lower bound on the number of summands, that is, the number of integer compositions $\ell_1+\ldots+\ell_m=n$.
	As we mentioned at the very end of \cref{sec:c_ideal}, we are not aware of strong approximation for this number and the following estimation is very crude for many values of $n$, $m$, and $c$:
	If we let the first $m- \frac{m}{c}$ integers, $\ell_1,\ldots,\ell_\frac{m}{c}$, range freely between $0$ and $\alpha$, we can always ensure that the condition $\ell_1+\ldots+\ell_m = n$ is satisfied by choosing the remaining $\frac{m}{c}$ summands $\ell_{\frac{m}{c}+1},\ldots,\ell_m$ appropriately.
	Therefore, the number of summands is at least $(\alpha+1)^{m(1-\frac1c)}$. 
	Together with the previous calculation, this leads to
	\begin{align*}
	\mathbb{E}[\chi_{\{Y=n\}}\chi_{\max\{Y_1,\ldots,Y_m\}\le d}] &\ge \frac{\alpha^n}{\euler^n} \frac{(\alpha+1)^{m(1-\frac1c)}}{(d!)^{m/c}}.
	\end{align*}
	Putting everything together and using the Stirling bounds~\cite{stirling} for the factorials, we obtain the desired result.
\end{proof}

Combining the results in this section in the straightforward way yields the following upper bound on $H_c$.
\begin{theorem}[Upper Bound on Minimal Family Size $\bm{H_c}$]\label{cor:upperbound_H_c}
	The minimal number of hash functions in a $c$-ideal family of hash functions is bounded from above by
	\[H_c \le \upperboundHc.\]
\end{theorem}
\begin{proof}
	We start with the probability bound from \cref{thm:probabilitybound2} and use \cref{thm:replacement2} and \cref{thm:upperbound}.
	\begin{align*}
	H_c&\le\left\lceil\frac {|\mathcal{S}_n|}{M_c}n\ln u\right\rceil\\
	\cmt{\Cref{thm:replacement2}}&\le \left\lceil\frac {1}{\mathop{\mathbb{P}}(T_\textnormal{max}\le c\alpha)}n\ln u\right\rceil\\
	\cmt{\Cref{thm:upperbound} and $d=c\alpha$}&\le\left\lceil
	\frac{(2\pi c\alpha)^{m/(2c)}}{\sqrt{2\pi n}}c^n\frac{\euler^{\frac{m}{12c^2\alpha}}}{(\alpha+1)^{m(1-\frac1c)}}
	n\ln u\right\rceil\\
	\cmt{$n=m\alpha$}&=\left\lceil
	\left((2\pi c\alpha)^{1/(2c)}\left(c^\alpha\right)\frac{\euler^{\frac{1}{12c^2\alpha}}}{(\alpha+1)^{1-\frac1c}}\right)^m
	\sqrt{\frac{n}{2 \pi}}\ln u\right\rceil
	\laterdef{\upperboundHc}{\left\lceil
		\left(\sqrt{2\pi c\alpha}^{1/c}c^\alpha\frac{\euler^{\frac{1}{12c^2\alpha}}}{(\alpha+1)^{1-\frac1c}}\right)^m
		\sqrt{\frac{n}{2 \pi}}\ln u\right\rceil}
	\laterdef{\upperboundHcOnotation}{\mathcal{O}\left(
		\ln{\ln{u}} + \ln{n} + m\left(\frac1{2c}\log(2\pi c\alpha)+\alpha\log(c)+\frac1{12c^2\alpha}-\left(1-\frac1c\right)\log(\alpha+1)
		\right)
		\right)}\\
	\end{align*}
\end{proof}
\Cref{cor:upperbound_H_c} states that for constant $\alpha$, the number of functions such that for each set of keys of size $n$, there is a function that distributes this set among the hash table cells at least as good as $c$ times an optimal solution is bounded from above by a number that grows exponentially in $m$, but only with the square root in $n$ and only logarithmically with the universe size.
However, our bounds do not match, hence the exact behavior of $H_c$ within the given bounds remains obscure.
It becomes easier if we analyze the advice complexity, using the connection between $c$-ideality and advice complexity described in the introduction of this paper.
We first improve our bounds on $H_c$ for some edge cases in the next section before we then use and interpret the results in \cref{sec:adv_comp}.

\section{Improvements For Edge Cases}\label{sec:edgecases}
As discussed in the introduction, the size $u$ of the universe does not appear in the lower bound on $H_c$ of~\cref{cor:mainbounds}. 
The following lower bound on $H_c$, which is a straightforward generalization of an argument presented in the classical textbook by Mehlhorn~\cite{mehlhorn}, features $u$ in a meaningful way.

\begin{theorem} \label{lowerboundthm1}
	We have that \[ H_c \ge \frac{\ln(u) - \ln(c\alpha)}{\ln(m)}.\]
\end{theorem}
\begin{proof}
	Denote the hash functions in the given family $\mathcal{H}$ by $h_1,\ldots,h_{|\mathcal{H}|}$, in an arbitrary order. Consider $h_1:U \to T$. Since it splits up $u$ keys among $m$ cells, there exists a cell $k_1\in T$ with at least $\frac{u}{m}$ keys, that is, $|h_1^{-1}(k)|\ge\frac um$.
	In other words, at least $\frac um$ of the keys are indistinguishable under the hash function $h_1$.
	Consider now $h_2:U \to T$ and its behavior on the keys from $h_1^{-1}(k)$. By the same argument as before, $h_2$ cannot do better but to split them up evenly among the $m$ cells, so there exists a cell $k_2\in T$ that is allotted at least an $m$th of these $|h_1^{-1}(k)|$ keys. Hence, $|h_1^{-1}(k_1)\cap h_2^{-1}(k_2)|\ge \frac{u}{m^2}$. Iterating this argument through all $|\mathcal{H}|$ functions in the given family, there is a set of at least
	\[ d\coloneqq \left|\bigcap_{i=1}^{|\mathcal{H}|}h_i^{-1}(k_i)\right|\ge \frac{u}{m^{|\mathcal{H}|}}\]
	keys in the universe that are indistinguishable under all functions in $\mathcal{H}$. By definition, $\mathcal{H}$ can only be $c$-ideal if $d\le c\alpha$. By transforming the equation above, this leads to
	\[|\mathcal{H}| \ge \frac{\ln(u)-\ln(c\alpha)}{\ln(m)}.\]	
\end{proof}

This demonstrates that, while it is easy to find bounds that include the size of the universe, it seems to be very difficult to incorporate $u$ into a general bounding technique that does not take it into account naturally, such as the first inequality of~\cref{cor:mainbounds}.
However, it is not too difficult to obtain bounds that improve upon~\cref{cor:mainbounds} for large---possibly less interesting---values of $c$.

\begin{theorem}[Yao Bound]\label{upperbound_yao}
	For every $c\in\omega\left(t\frac{\ln n}{\ln\ln n}\right)$, where $t\ge 1$, we have
	$H_c\in\mathcal{O}\left(\ln |\mathcal{S}_n| /\ln t \right)$. In particular, we obtain that $H_c\in\mathcal{O}\left(n\ln u\right)$ for $t\in\mathcal{O}(1)$.
	\laterdef{\upperboundHcYao}{\frac{\ln |\mathcal{S}_n|}{\ln t}}
\end{theorem}

\begin{proof}
	We apply the Yao-inspired principle described by Komm~\cite{Komm}, which can be summarized as follows:
	Suppose that a hash function $h\in\mathcal{H}_\text{all}$ is chosen uniformly at random. The expected maximum cell load is $\mu\coloneqq\mathbb{E}[\alpha_\textnormal{max}]\in\Theta\left(\frac{\ln n}{\ln\ln n}\right)$, as proven in~\cite{ballsintobins}.
	The Markov inequality guarantees that at most a fraction $1/t$ of inputs result in a cell load of $t\mu$ or more:
	\[\mathop{\mathbb{P}}\limits_{S\in\mathcal{S}_n}[\alpha_\textnormal{max}\ge t\mu]\le \frac1t.\]
	Therefore, we can find a single hash function $h$ that guarantees a maximal load of at most $t\mu$ for all but a fraction of $1/t$ of all input sets $S\in\mathcal{S}_n$ simultaneously. 
	Now we can repeat the argument as many times as we desire and will find every time a hash function that causes a maximal load of $t\mu$ or more for at most a fraction $1/t$ of inputs for which the previous hash functions did so. 
	After $r$ iterations, there are only $|\mathcal{S}_n|t^{-r}$ input sets left that cause a maximal load of at least $t\mu$. 
	Seeing that 
	\[|\mathcal{S}_n|\left(\frac1t\right)^r< 1\ \Longleftrightarrow\ |\mathcal{S}_n|< t^r\ \Longleftrightarrow\ \frac{\ln|\mathcal{S}_n|}{\ln t}< r,\]
	we can conclude that there is a subfamily of size at most $\lfloor\ln|\mathcal{S}_n|/\ln t\rfloor + 1$ that contains, for every subset $S\in\mathcal{S}_n$, at least one hash function that incurs cell loads of at most $t\mu$.
\end{proof}

\section{Advice Complexity of Hashing}\label{sec:adv_comp}
With the conceptualization of hashing as an ultimate online problem mentioned in the introduction of this paper, we can use the bounds on $H_c$ to provide bounds on the advice complexity of $c$-competitive algorithms.
\Cref{lowerboundthm1} immediately yields the following theorem.
\begin{theorem}\label{cor:adviceeasylowerbound}
	Every \algA for hashing with less than 
	$\ln\left(\ln(u)-\ln(c\alpha)\right)-\ln(\ln(m))$
	advice bits cannot achieve a lower cost than $\cost(\algA) = c\alpha$ and is thus not better than $c$-competitive. In other words, there exists an $S \subseteq U$ such that the output $h:U \to T$ of \algA maps at least $c\alpha$ elements of $S$ to one cell.
\end{theorem}

With a lower bound on the size of $c$-ideal families of hash functions, this bound can be improved significantly, as the following theorem shows.

\begin{theorem}\label{cor:advicecideallowerbound}
	Every \algA for hashing needs at least
	\[\log\left(\lowerboundHc\right)\]
	advice bits in order to be $c$-competitive, for an arbitrary fixed $\varepsilon > 0$.
\end{theorem}
We want to determine the asymptotic behavior of this bound and hence analyze the term $(\alpha/(c\alpha+1))^{c\alpha+1}$.
We are going to use the fact that $\lim_{n\to \infty}(1-1/n)^n = 1/\euler$.
\begin{align*}
	\left(\frac{\alpha}{c\alpha+1}\right)^{c\alpha+1} &= \left(\frac{\frac 1 c (c\alpha + 1 -1)}{c\alpha+1}\right)^{c\alpha+1} \\
	&= \left(\frac 1 c \left(1-\frac{1}{c\alpha+1}\right)\right)^{c\alpha+1} \\
	\cmt{for any $\varepsilon'>0$ and $c\alpha+1$ large enough} &\ge \left(\frac 1 c\right)^{c\alpha+1}(1-\varepsilon')\frac 1 \euler
\end{align*}
Therefore, the bound from \cref{cor:advicecideallowerbound} is in
\[ \Omega\left( \frac{m}{\euler^{\alpha}}\left(\frac{1}{c}\right)^{c\alpha+1}\right).\]
Before we interpret this result, we turn to the upper bounds. 
\Cref{cor:upperbound_H_c} yields the following result.

\begin{theorem}\label{thm:advicecidealupperbound}
	There is a $c$-competitive algorithm with advice that reads
	\[\log\upperboundHc\]
	\[\in \upperboundHcOnotation\]
	advice bits.
\end{theorem}
The factor after $m$ is minimal for $c=\alpha=1$; this minimal value is larger than 1.002.
This upper bound is therefore always at least linear in $m$.
For $c = \omega\left(\frac{\ln(n)}{\ln(\ln(n))}\right)$, we can improve on this and remove the last summand completely, based on \cref{upperbound_yao}.

\begin{theorem}\label{thm:adviceYaoupperbound}
	For $c\in\omega\left(t\frac{\ln n}{\ln\ln n}\right)$, $t> 1$, there exists a $c$-competitive algorithm that reads
	$\mathcal{O}\left(\log\left(\upperboundHcYao\right)\right)$
	many advice bits. In particular, for $t\in\mathcal{O}(1)$, there exists an $\frac{\ln n}{\ln\ln n}$-competitive algorithm that reads $\mathcal{O}\left(\ln\ln u + \ln n \right)$
	many advice bits.
\end{theorem}
\Cref{cor:advicecideallowerbound} and \cref{thm:advicecidealupperbound} reveal the advice complexity of hashing to be linear in the hash table size $m$. 
While the universe size $u$ still appears in the upper bound, it functions merely as a summand and is mitigated by a double logarithm. 
Unless the key length $\log_2(u)$ is exponentially larger than the hash table size $m$ ,
the universe size cannot significantly affect this general behavior.
The more immediate bounds for the edge cases do not reveal this dominance of the hash table over the universe.
Moreover, changing the two parameters $\alpha$ and $c$ has no discernible effect on the edge case bounds despite the exponential influence on the main bounds.

\section{Conclusion} \label{hash:sec:conclusion}
This paper analyzed hashing from an unusual angle by regarding hashing as an online problem and then studying its advice complexity.
Online problems are usually studied with competitive analysis, which is a worst-case measurement.
As outlined in the introduction, 
it is impossible to prevent a deterministic algorithm 
from incurring the worst-case cost by hashing all appearing 
keys into a single cell of the hash table.
Therefore, randomized algorithms are key to the theory and application of hashing.
In particular the surprising discovery of small universal hashing families gave rise to efficient algorithms with excellent expected cost behavior.
However, from a worst-case perspective, the performance of randomized algorithms is lacking.

This motivated the conceptualization of $c$-ideal hashing families as a generalization of perfect $k$-hashing families to the case where $\alpha>1$.
Our goal was to analyze the trade-off between size and ideality of hashing families since this is directly linked to the competitiveness of online algorithms with advice.
Our bounds generalize results by Fredman and Komlós~\cite{fredman} as well as Naor et al.~\cite{splitters} to the case $\alpha>1$ and $c\ge 1$.

As a first step, 
we proved that balanced hash functions 
are suited best for hashing in the sense that they 
maximize the number of subsets that are hashed $c$-ideally.
Building on this, we applied results by Dubhashi and Ranjan~\cite{negativedependence}
to obtain our main lower bound of (\ref{equation_mainlowerbound}). 
Our second lower bound, (\ref{equation_easylowerbound}), is a straightforward generalization of a direct approach for the special case $c=\alpha=1$ by Mehlhorn~\cite{mehlhorn}.
We used two techniques to find complementing upper bounds.
The first upper bound, (\ref{equation_mainupperbound}), uses a Poissonization method combined with direct calculations and is mainly useful for $c\in o\left(\frac{\ln m}{\ln\ln m}\right)$.
Our second upper bound, (\ref{equation_Yaobound}), relies on a Yao-inspired principle~\cite{Komm} and covers the case $c\in \omega\left(\frac{\ln m}{\ln\ln m}\right)$.

With these results on the size of $c$-ideal hash function families, we discovered that the advice complexity of hashing is linear in the hash table size $m$ and only logarithmic in $n$ and double logarithmic in $u$ (see Schmidt and Siegel \cite{schmidt} for similar results for perfect hashing).
Moreover, the influence of both $\alpha$ and $c$ is exponential in the lower bound. In this sense, by relaxing the pursuit of perfection only slightly, the gain in the decrease of the size of a $c$-ideal hash function family can be exponential.
Furthermore, only $O(\ln\ln u + \ln n)$ advice bits are necessary for deterministic algorithms to catch up with randomized algorithms.

Further research is necessary to close the gap between our upper and lower bounds.
For the edge cases, that is, for $c \ge \log{n}/\log{\log{n}}$, the upper and lower bounds (\ref{equation_Yaobound}) and (\ref{equation_easylowerbound}) differ by a factor of approximately $n\ln m$.
The interesting case for us, however, is $c \le \log{n}/\log{\log{n}}$, which 
is the observed worst-case cost for universal hashing.
Contrasting (\ref{equation_mainlowerbound}) and (\ref{equation_mainupperbound}), 
we note that the difference has two main reasons. First, there is a factor of $n\ln u$ that appears only in the upper bound; this factor stems from the general bounds in \cref{cor:mainbounds}.
Second, the probability $\P(\alpha_\textnormal{max} \le c\alpha)$ is estimated from below in a more direct fashion than from above, leading to a difference between these bounds that increases with growing $c$.
The reason for the more direct approach is the lack of a result similar to \cref{lem:negativeregression}. 

Moreover, it remains an open question whether it is possible to adapt the entropy-related methods in the spirit of Fredman-Komlós and Körner in such a way as to improve our general lower bound (\ref{equation_mainlowerbound}) by accounting for the universe size in a meaningful way.

\bibliographystyle{plainurl}
\bibliography{bib}

\end{document}